\documentclass[letterpaper, 10 pt, journal, twoside]{IEEEtran}

\usepackage{amsmath} 
\usepackage{amssymb}  
\usepackage{graphicx}
\usepackage{epstopdf}
\usepackage{amsmath}
\usepackage{mathtools}
\usepackage{dsfont}
\usepackage{tikz}
\usepackage{siunitx}
\usepackage{xcolor}
\usepackage[bottom]{footmisc} 

\usepackage{hyperref}

\usepackage{balance}    
\usepackage{amsthm}

\usepackage{balance}    
\usepackage{amsthm}
\usepackage{xcolor}
\usepackage{tcolorbox}
\usepackage{algpseudocode}
\usepackage{dsfont}
\usepackage{algorithm}

\def\sq{\mathbin{{\strut\rule{1.25ex}{1.25ex}}}}

\newcommand{\beq}{\begin{equation}}
\newcommand{\eeq}{\end{equation}}

\def\mathcolor#1#{\@mathcolor{#1}}
\def\@mathcolor#1#2#3{%
  \protect\leavevmode
  \begingroup
    \color#1{#2}#3%
  \endgroup
}

\newcounter{algorithmctr}[section]
\renewcommand{\thealgorithmctr}{\thesection.\arabic{algorithmctr}}
{\refstepcounter{algorithmctr}\begin{list}{}{%
\setlength{\rightmargin}{0\linewidth}%
\setlength{\leftmargin}{.05\linewidth}
\setlength{\itemsep}{1pt}
\setlength{\parskip}{0pt}
\setlength{\parsep}{0pt}}%
\rmfamily\small
\item[]{\setlength{\parskip}{0ex}\hrulefill\par%
\nopagebreak{\bfseries\textsf{Algorithm \thealgorithmctr~}}}}%
{{\setlength{\parskip}{-1ex}\nopagebreak\par\hrulefill} \end{list}}
\IEEEoverridecommandlockouts

\usepackage[noadjust]{cite} 
\usepackage{amsmath,stmaryrd,graphicx}

\newtheoremstyle{boldStyle}
  {\topsep}
  {\topsep}
  {\itshape}
  {0pt}
  {\bfseries}
  {.}
  { }
  {\thmname{#1}\thmnumber{ #2}\thmnote{ (#3)}}

\newtheoremstyle{italicStyle}
  {\topsep}
  {\topsep}
  {}
  {0pt}
  {\bfseries}
  {.}
  { }
  {\thmname{#1}\thmnumber{ #2}\thmnote{ (#3)}}

\theoremstyle{boldStyle}

\newtheorem{theorem}{Theorem}

\newtheorem{corollary}{Corollary}

\theoremstyle{italicStyle}
\newtheorem{assumption}{Assumption}
\newtheorem{remark}{Remark}

\usepackage{nopageno}

\renewenvironment{proof}{{\textbf{Proof:}}}{\hfill$\sq$}

\makeatletter
\newcommand{\fixed@sra}{$\vrule height 2\fontdimen22\textfont2 width 0pt\shortrightarrow$}
\newcommand{\shortarrow}[1]{%
  \mathrel{\text{\rotatebox[origin=c]{\numexpr#1*45}{\fixed@sra}}}
}
\newcommand\fs@betterruled{%
  \def\@fs@cfont{\bfseries}\let\@fs@capt\floatc@ruled
  \def\@fs@pre{\vspace*{5pt}\hrule height.8pt depth0pt \kern2pt}%
  \def\@fs@post{\kern2pt\hrule\relax}%
  \def\@fs@mid{\kern2pt\hrule\kern2pt}%
  \let\@fs@iftopcapt\iftrue}
\floatstyle{betterruled}
\restylefloat{algorithm}
\makeatother

\title{\LARGE \bf 
Iterative Model Predictive Control for Piecewise Systems}

\author{ Ugo Rosolia and Aaron D. Ames  
\thanks{Ugo Rosolia and Aaron D. Ames are with the AMBER lab at Caltech, Pasadena, USA. E-mails: {\tt\scriptsize{\{urosolia, ames\}@caltech.edu}}. The authors would like to acknowledge the support of the NSF award \#1932091
. }
}%

\begin{document}

\maketitle
\thispagestyle{empty}
\pagestyle{empty}
\begin{abstract}
In this paper, we present an iterative Model Predictive Control (MPC) design for piecewise nonlinear systems. We consider finite time control tasks where the goal of the controller is to steer the system from a starting configuration to a goal state while minimizing a cost function. First, we present an algorithm that leverages a feasible trajectory that completes the task to construct a control policy which guarantees that state and input constraints are recursively satisfied and that the closed-loop system reaches the goal state in finite time. Utilizing this construction, we present a policy iteration scheme that iteratively generates safe trajectories which have non-decreasing performance. Finally, we test the proposed strategy on a discretized Spring Loaded Inverted Pendulum (SLIP) model with massless legs. We show that our methodology is robust to changes in initial conditions and disturbances acting on the system. Furthermore, we demonstrate the effectiveness of our policy iteration algorithm in a minimum time control task. 
\end{abstract}

\section{Introduction}

Robots performing complex tasks can be described as hybrid systems, which are characterized by continuous dynamics and discrete events. Therefore, controllers designed for such systems can take control actions based on continuous and discrete decision variables. Yet the presence of discrete variables make planning and control problems challenging, as it is required to reason about all possible combinations of discrete events. This challenge can be mitigated by designing hierarchical strategies, where a high-level planner computes the discrete variables and a low-level controller optimizes the system trajectory described by continuous variables~\cite{jenelten2020perceptive, villarreal2020mpc, grandia2020multi, grandia2020nonlinear}. 

A popular methodology to synthesize policies, which can jointly plan over discrete and continuous states is Model Predictive Control (MPC)~\cite{bemporad1999control, bemporad2000piecewise, oberdieck2015explicit, borrelli2017predictive}. MPC is a control strategy which systematically uses forecast to compute control actions. At each time step, an MPC plans a trajectory over a short time window, then the first control action is applied to the system and the process is repeated at the next time step based on new measurements. When the system dynamics are hybrid, the MPC planning problem is a Mixed Integer Program (MIP) that is hard to solve online with limited computational resources. For this reason, significant work has focused on explicit MPC strategies where the solution to the MIP is solved offline as a parametric optimization problem~\cite{dua2002multiparametric, oberdieck2015explicit, borrelli2017predictive}. For hybrid systems described by piecewise affine dynamics the parametric optimization problem can be solved exactly. Once the solution is computed offline, the MPC policy is given by a look-up table of feedback gains that can be efficiently implemented online in real-time~\cite{nascu2017explicit, darivianakis2014hybrid}. 
However, computing the explicit solution to hybrid MPC problems is computationally demanding.

Another strategy to speed-up the computation of the MPC policy is to leverage warm-starting strategies, where the optimization algorithm is initialized using a candidate solution. 
Several strategies have been proposed for warm-starting hybrid MPC problems~\cite{marcucci2017approximate, marcucci2020warm, frick2015embedded, hespanhol2019structure}. These approaches leverage the  trajectory computed at the previous time step to warm-start both the continuous and discrete variables. As the complexity of solving MIPs is given by the computation of the optimal integer variables, recent works have investigated the possibility of leveraging learning algorithms to predict the set of active discrete variables used to warm-start the MPC~\cite{zhu2019fast, bertsimas2020voice, agrawal2020learning}. 

In this work, we focus on control tasks where the goal is to steer the system from a starting configuration to a goal state in finite time, while satisfying state and input constraints. We assume that a feasible trajectory that is able to perform the task is available. Then, we synthesize a control policy, which plans the system trajectory over a finite horizon that is shorter than the control task duration and may cause the controller to take unsafe shortsighted control actions. Thus, building upon~\cite{rosolia2017learning, rosolia2020minimum}, we present a methodology to construct the MPC terminal components in order to guarantee satisfaction of the safety constraints and convergence in finite time of the closed-loop system to the goal set.

Compared to previous works~\cite{rosolia2017learning, rosolia2020minimum}, we show how to handle piecewise systems by warm-starting the integer variables and we present a shrinking horizon strategy tailored to finite time control tasks. We present an algorithm which solves at most $M$ smooth optimization problems and is guaranteed to find a feasible solution to the original MIP planning problem. Our approach is based on a sub-optimal trajectory that can complete the task and affects the closed-loop performance of the controller. Therefore, we present a policy iteration algorithm, where simulations are used to iteratively update the controller. We prove that our algorithm returns safe policies that have non-decreasing performance. Finally, we demonstrate the effectiveness of our approach on the discretized Spring Loaded Inverted Pendulum (SLIP)~\cite{shahbazi2016unified}. 


\section{Problem Formulation}\label{sec:problemFormulation}
In this section, we describe the system model and the control synthesis objectives. We consider discrete time piecewise nonlinear systems defined over $R$ disjoint regions  $\mathcal{D}_i \subseteq \mathbb{R}^n$ for $i \in \{1,\ldots, R\}$:
\begin{equation}\label{eq:sys}
    x_{t+1} = f_i(x_t, u_t), \quad \mbox{if } x_t \in \mathcal{D}_i,
\end{equation}
where the state $x_t \in \mathbb{R}^n$ and the input $u_t \in \mathbb{R}^d$. In the above equation $f_i : \mathbb{R}^n \times \mathbb{R}^d \rightarrow \mathbb{R}^n$ represents the system dynamics, which describe the evolution of the discrete time system when the state $x_t$ belongs to the region $\mathcal{D}_i \subseteq \mathbb{R}^n$. 
Furthermore, the system is subject to the following state and input constraints: 
\begin{equation}\label{eq:cnstr}
    u_t \in \mathcal{U} \subseteq \mathbb{R}^d \text{ and } x_t \in \mathcal{X} \subseteq \mathbb{R}^n,~\forall t \in \{0, \ldots, T-1\},
\end{equation}
where $T \in \{0,1,\ldots\}$ is the duration of the control task. Notice that several robotic systems can be described by piecewise constrained nonlinear models defined over disjoint regions, 
such as the SLIP model presented in Section~\ref{sec:SLIP}. 

\textbf{Objective: } Our goal is to design a control policy $\pi:\mathbb{R}^n \rightarrow \mathbb{R}^d$ which maps states to actions, i.e.,
\begin{equation}\label{eq:policy}
    u_t = \pi(x_t).
\end{equation}
The above control policy should guarantee that the state and input constraints from~\eqref{eq:cnstr} are satisfied and that the closed-loop system~\eqref{eq:sys} and \eqref{eq:policy} converges in finite time to a goal set $\mathcal{G} \subset \mathcal{X}$. More formally, the control policy~\eqref{eq:policy} should guarantee that for an initial condition $x_I$ in a neighborhood of a starting state $x_S \in \mathcal{X}$, the trajectory of the closed-loop system~\eqref{eq:sys} and \eqref{eq:policy} is a feasible solution to the following \textit{Finite Time Optimal Control Problem (FTOCP)}:
\begin{equation}\label{eq:ftocp}
    \begin{aligned}
        \min_{\substack{u_0, \ldots, u_{T-1}\\ i_0, \ldots, i_{T-1} } } \quad & \sum_{t=0}^{T-1} l(x_t, u_t) \\
        \text{s.t.} ~\quad \quad & x_{t+1} = f_{i_t}(x_t, u_t),\\
        & x_t \in \mathcal{D}_{i_t} \cap \mathcal{X} ,~u_t \in \mathcal{U},~i_t \in \{1,\ldots,R\},\\
        & x_0 = x_I, x_T \in \mathcal{G},\\
        & \forall t\in\{0,\ldots,T-1\},
    \end{aligned}
\end{equation}
where the stage cost $l:\mathbb{R}^n \times \mathbb{R}^d \rightarrow \mathbb{R}$. 
Notice that in the above problem the system dynamics are a function of the integer variables $i_t \in \{1, \ldots, R\}$. Therefore, for a feasible set of continuous inputs $[u_0, \ldots, u_{T-1}]$ and integer variables $[i_0, \ldots, i_{T-1}]$, we have that the resulting vector of states $[x_0, \ldots, x_{T}]$ must satisfy $x_t \in \mathcal{D}_{i_t}\cap\mathcal{X},~\forall t \in \{0,\ldots,T-1\}$. Throughout the paper we make the following assumptions.



\begin{assumption}\label{ass:firstFeas}
We are given the state-input trajectories which are feasible for the FTOCP~\eqref{eq:ftocp} with $x_I = x_S$: 
\begin{equation*}\label{eq:firstFeas}
    {\mathbf{x}}^0 = [x_0^0, \ldots, x_{T}^0] \text{ and } {\mathbf{u}}^0 = [u_0^0, \ldots, u_{T-1}^0],
\end{equation*}
where for all $t \in \{0,\ldots,T-1\}$ the state $x_t^0 \in \mathcal{X}$, the input $u_t^0 \in \mathcal{U}$ and $x_{T}^0\in\mathcal{G}$. 
\end{assumption}

\begin{assumption}\label{ass:invariance}
For any $x\in\mathcal{G}$, the stage cost $l(x,u)=0$, for all $u \in \mathcal{U}$. Moreover, the set $\mathcal{G}$ is control invariant, i.e, for all $x \in \mathcal{G}$, there exists a control $u \in \mathcal{U}$ and index $i \in \{1,\ldots,R\}$ such that $f_i(x,u) \in \mathcal{G}$ and $ x\in\mathcal{D}_i$.
\end{assumption}

\begin{remark}
The proposed methodology requires only feasibility of a trajectory ${\mathbf{x}}^0$. However, the optimality of the trajectory ${\mathbf{x}}^0$ affects the performance of the proposed control synthesis strategy. For this reason, in Section~\ref{sec:iterativeDesign} we present a policy iteration scheme that may be used to iteratively improve
 the closed-loop performance of the policy.

\end{remark}

\section{Control Design}\label{sec:controller}
In this section, we first introduce an FTOCP which can be recast as a non-linear program (NLP). Afterwards, we present the proposed strategy which leverages this NLP and the feasible trajectory from Assumption~\ref{ass:firstFeas}. Finally, we present a policy iteration scheme which can be used to improve the performance of the control policy.

\subsection{Model Predictive Control}
The FTOCP~\eqref{eq:ftocp} is challenging to solve as at each time $t$ the system dynamics change as a function of the state $x_t$, i.e., $x_{t+1} = f_i(x_t, u_t)$ if $x_t \in \mathcal{D}_i$. 
However, the computational complexity may be reduced when searching for a feasible sub-optimal solution. 
In particular, a feasible solution to the FTOCP~\eqref{eq:ftocp} may be computed fixing a priori a sequence of regions $\{\mathcal{D}_{i_0},\ldots,\mathcal{D}_{i_{T-1}}\}$, where the system should be constrained at each time $t$. 
Clearly, a trajectory which steers the system from the starting state $x_I$ to the goal set $\mathcal{G}$, while visiting the sequence of regions $\{\mathcal{D}_{i_0},\ldots,\mathcal{D}_{i_{T-1}}\}$ is a feasible trajectory for the original FTOCP~\eqref{eq:ftocp}.

In order to reduce the computational complexity, we introduce an FTOCP defined over a horizon $N$ shorter than $T$ and for a set of indices $\mathcal{I}_t=\{i_t, \ldots, i_{t+N-1}\}$ associated with a sequence of $N$ regions $\{\mathcal{D}_{i_t}, \ldots, \mathcal{D}_{i_{t+N-1}}\}$.
In particular, given a set of indices $\mathcal{I}_t$, the terminal state $x_F$ and an associated terminal cost $q_F\in\mathbb{R}$ we define the FTOCP:
\begin{equation}\label{eq:mpc}
    \begin{aligned}
        J(x_t, x_F, q_F, \mathcal{I}_t, N) = \min_{\mathbf{u}_t} \quad &\sum_{k=t}^{t+N-1} l(x_{k|t}, u_{k|t}) +  q_F \\
        \text{s.t.}  ~~~ & x_{k+1|t} = f_{i_k}(x_{k|t}, u_{k|t}),\\
        & x_{k|t} \in \mathcal{D}_{i_k} \cap \mathcal{X},~u_{k|t} \in \mathcal{U}, \\
        & x_{t|t} = x_t, x_{t+N|t} = x_F,\\
        &\forall k \in \{t, \ldots, t+N-1\},
    \end{aligned}
\end{equation}
where $\mathbf{u}_t=[u_{t|t}, \ldots, u_{t+N-1|t}]$
The optimal state-input sequences to the above FTOCP
\begin{equation}\label{eq:opt}
    [x_{t|t}^*, \ldots, x_{t+N|t}^*] \text { and }[u_{t|t}^*, \ldots, u_{t+N-1|t}^*],
\end{equation}
steer the system from the starting state $x_t$ to the terminal state $x_F$ while satisfying state and input constraints. 

In the FTOCP~\eqref{eq:mpc}, at each predicted time $k$ the system state $x_{k|t} \in \mathcal{D}_{i_k} \cap \mathcal{X}$, and therefore problem~\eqref{eq:mpc} can be recast as an NLP, which is easier to solve than problem~\eqref{eq:ftocp} where the optimization is carried out over continuous and integer variables. 
Next, we present the proposed algorithm which chooses the set $\mathcal{I}_t=\{i_t, \ldots, i_{t+N-1}\}$ associated with the sequence of regions $\{\mathcal{D}_{i_t}, \ldots, \mathcal{D}_{i_{t+N-1}}\}$,
the terminal state $x_F$, and terminal cost $q_F$ that are used in the FTOCP~\eqref{eq:mpc}.

\begin{algorithm}[t!]
  \caption{Control Policy $\pi$}\label{algo:policy}
  \begin{algorithmic}[1]
      \State \textbf{Init Parameters:} ${\mathbf{q}}^0$, ${\mathbf{x}}^0$, ${\mathbf{i}}^0$, $M$, $k_0 = N$, $N_0=N$, $T$
      \State \textbf{Input:} $x_t$
    \State $i_t = \texttt{getRegion}(x_t)$
      \For{$m = [0, \ldots, M-1]$} \Comment{Solve $M$ FTOCPs}
        \State set $t_F = \min(k_t + m, T)$
        \State set $x_F = x^0_{t_F}$ 
        \State set $q_F = q^0_{t_F}$
        \State set $\mathcal{I}_t = \{i_t,  i^0_{t_F-N_t+1},\ldots, i^0_{t_F-1}  \}$
        \State solve the FTOCP $J(x_t, x_F, q_F, \mathcal{I}_t, N_t)$ from~\eqref{eq:mpc}
        \State store $c_m = J(x_t, x_F, q_F, \mathcal{I}_t, N_t)$
        \State store $\bar u_m = u^*_{t|t}$ 
        \State store $\bar x_{F,m} = x^*_{t+N_t|t}$
        \If {$m>0$ and $c_{m-1} < c_m$} 
          \State $c^*_t = c_{m-1}$ \Comment{Pick best cost}
          \State $m^*_t = m-1$ \Comment{Pick best cost index}
        \State \textbf{break}
        \EndIf
      \EndFor
        \If {$\bar x_{F,m^*_t} = x_{T}^0$} \Comment{Horizon and Parameter Update}
        \State store $N_{t+1} = \max(1, N_t-1)$
      \State store $k_{t+1} = T$
	\Else
        \State store $N_{t+1} = N_t$
      \State store $k_{t+1} = k_t + m^*_t+1$ 
        \EndIf
      \State $u_t = \bar u_{m^*_t}$
      \State \textbf{Outputs} $c_t^*$, $u_t$
  \end{algorithmic}
\end{algorithm}

\subsection{Policy Synthesis}
This section describes the proposed strategy. For each state $x_t^0$ of the feasible trajectory $\mathbf{x}^0$ from Assumption~\ref{ass:firstFeas}, we define the vector
\begin{equation*}
    {\mathbf{i}}^0 = [i_0^0, \ldots, i_T^0],
\end{equation*}
where $i_t^0 \in \{1, \ldots, R\}$ identifies the region containing the system's state at time $t$, i.e., $x_t^0 \in \mathcal{D}_{i_t^0}$.
Moreover, for each state $x_t^0$ we introduce the cost-to-go $q_t^0$ given by the recursion:
\begin{equation}\label{eq:realizedCost}
q_t^0 = l(x_k^0, u_k^0)+q_{t+1}^0,
\end{equation}
for $q_{T}^0=0$ and the stage cost $l :\mathbb{R}^n \times \mathbb{R}^d \rightarrow \mathbb{R}$ from~\eqref{eq:ftocp}.
Finally, we define the cost vector 
\begin{equation*}
    {\mathbf{q}}^0 = [q_0^0, \ldots, q_{T}^0].
\end{equation*}

The cost vector ${\mathbf{q}}^0$,  the feasible trajectory ${\mathbf{x}}^0$ and the vector of indices ${\mathbf{i}}^0$, together with the parameters $M\in\{0,1,2,\ldots\}$ and $N\in\{0,1,\ldots\}$, are used to initialize the proposed control policy, which is described by Algorithm~\ref{algo:policy}. At each time $t$, Algorithm~\ref{algo:policy} takes as input the state of the system $x_t$ and it returns the best cost and the control action~$u_t$, which is applied to system~\eqref{eq:sys}. First, given the state $x_t$ we identify the region's index $i_t$ such that $x_t \in \mathcal{D}_{i_t}$ (line~1). Afterwards, we solve $M$ times the FTOCP~\eqref{eq:mpc} with a different terminal state $x_F = x^0_{t_F}$, terminal cost $q_F = q^0_{t_F}$ and set of indices $\mathcal{I}_t$~(lines~4-18). Note that the set of indices $\mathcal{I}_t$ (line~8) is computed appending to the current region's index $i_t$ a sequence of $N_t-1$ indices selected backward from time $t_F$, which is chosen independently of $N_t$ for $t > 0$. As a result, at each time $t$ the controller solves the FTOCP~\eqref{eq:ftocp} using a sequence of indices which is different from the one associated with the first feasible trajectory\footnote{For example, at time $t=0$, for $m=2$, $k_0= N = 4$ and $T=100$, we have that $t_F = k_0+m=6$; therefore $\mathcal{I}_0 = \{i_0, i^0_{3}, i^0_{4}, i^0_{5}  \}$ and $x_F = x_6^0$.}. As we will discuss in Section~\ref{sec:properties} at time $t$ and for $m=0$, the FTOCP~\eqref{eq:mpc} is feasible when the terminal components are defined as in lines~5-7. Therefore if for $m>0$ we have that $c_{m-1}<c_m$, we stop solving the set of $M$ NLPs and we save the best cost and the associated index $m_t^*$. Then, 
we update the parameter $k_t$ that is used to define the terminal MPC components and we shrink the prediction horizon, if the terminal predicted state $\bar x_{F,m^*_t}$ equals the terminal state $x_{T}^0$. Finally, we return the optimal control action $u_t = \bar u_{m^*_t}$.

\vspace{-0.2cm}
\subsection{Policy Iteration}\label{sec:iterativeDesign}
The control policy from Algorithm~\ref{algo:policy} leverages the feasible trajectory~$\mathbf{x}^0$ to compute the terminal components used in the MPC problem~\eqref{eq:mpc} and, as a result, the performance of the proposed methodology is affected by the optimality of~$\mathbf{x}^0$. 
In this section, we discuss a policy iteration strategy that may be used to improve the performance of the control policy from Algorithm~\ref{algo:policy}. 
In particular, we simulate the closed-loop system and we iteratively update the control policy.

At iteration $j\geq1$, we define the control policy 
$\pi^j : \mathbb{R}^n \rightarrow \mathbb{R}^d$
which is given by Algorithm~\ref{algo:policy} initialized using the feasible trajectory ${\mathbf{x}}^{j-1}$, the vector of indices ${\mathbf{i}}^{j-1}$, and the cost vector ${\mathbf{q}}^{j-1}$. For an initial condition $x_0^j = x_S$, the policy~$\pi^j$ may be used to simulate the trajectory of the closed-loop system:
\begin{equation}\label{eq:sys_j}
x^{j}_{t+1} = f_i(x_t^j , \pi^j(x_t^j)), \quad \mbox{if } x_t^j \in \mathcal{D}_i,
\end{equation}
which is then used to update the policy $\pi^{j+1}$ at the next iteration $j+1$, as shown in Algorithm~\ref{algo:iterativeUpdate}. In what follows, we show that the closed-loop performance of the control policy~$\pi^j$ is non-decreasing at each $j$th policy update.

\begin{algorithm}[t!]
  \caption{Iterative Policy Update}\label{algo:iterativeUpdate}
  \begin{algorithmic}[1]
      \State \textbf{Init Parameters:} ${\mathbf{q}}^0$, ${\mathbf{x}}^0$, ${\mathbf{i}}^0$, $M$, $N$, $x_S$, $T$
      \State \textbf{Input:} $j$
      \State define $\pi^1$ via Algorithm~\ref{algo:policy} initialized with ${\mathbf{q}}^0$, ${\mathbf{x}}^0$, ${\mathbf{i}}^0$, $\phantom{sssssssssssdsdsdsdssdsdsd}$ $M$, $k_0=N$, $N_0=N$, $T$
      \For{$i \in \{1, \ldots, j\}$} \Comment{Policy iteration loop}
        \State simulate the closed-loop system~\eqref{eq:sys_j} from $x_0^i=x_S$
        \State set $\mathbf{x}^i=[x_0^i, \ldots, x_T^i]$
        \State compute the region indices $\mathbf{i}^i$ from $\mathbf{x}^i$
        \State compute $\mathbf{q}^i$ from~\eqref{eq:realizedCost} with $q_{T}^i=0$
          \State define $\pi^{i+1}$ via Algorithm~\ref{algo:policy} initialized with ${\mathbf{q}}^i$, ${\mathbf{x}}^i$, ${\mathbf{i}}^i$, $\phantom{sssssssssssdsdsdsdsdsdsdsdds}$ $M$, $k_0=N$, $N_0=N$
      \EndFor
      \State \textbf{Outputs} $(\pi^0, {\mathbf{x}}^0, {\mathbf{q}}^0), \cdots, (\pi^{j+1}, {\mathbf{x}}^{j+1}, {\mathbf{q}}^{j+1})$
  \end{algorithmic}
\end{algorithm}

\section{Properties}\label{sec:properties}

\subsection{Recursive Feasibility and Finite-Time Convergence}
We show that at each time $t$ Algorithm~\ref{algo:policy} in closed-loop with system~\eqref{eq:sys} guarantees constraint satisfaction and that the closed-loop system converges in finite time to the set $\mathcal{G}$.
\begin{theorem}\label{th:recFeas}
Consider the closed-loop system~\eqref{eq:sys} and~\eqref{eq:policy}, where the policy~$\pi$ is given by Algorithm~\ref{algo:policy}.
Let Assumptions~\ref{ass:firstFeas}-\ref{ass:invariance} hold. If at time $t=0$ the initial condition $x_0 = x_S$, then the closed-loop system~\eqref{eq:sys} and~\eqref{eq:policy} satisfies state and input constraints from~\eqref{eq:cnstr} and it reaches the terminal set $\mathcal{G}$ at time $T$, i.e., $x_t \in \mathcal{X},u_t \in \mathcal{U}, \forall t \in \{0,  \ldots, T-1\} \text{ and } x_T \in \mathcal{G}.$
\end{theorem}
\begin{proof}
The proof follows by standard MPC arguments~\cite{borrelli2017predictive, rawlings2009model} and the construction of the time-varying component from Algorithm~\ref{algo:policy}. Assume that at time $t$ Algorithm~\ref{algo:policy} returns a feasible control action $u_t$ and let
\begin{equation}\label{eq:optSeqProof}
    [x_{t|t}^*, \ldots, x_{t+N_t|t}^*=x_{F}] \text { and }[u_{t|t}^*, \ldots, u_{t+N_t-1|t}^*],
\end{equation}
be the optimal solution associated with the $m^*_t$th FTOCP~\eqref{eq:mpc} solved at time $t$. 
Next, we consider three cases to show that the time-varying components defining the FTOCPs solved at line~9 of Algorithm~\ref{algo:policy} guarantee that Algorithm~\ref{algo:policy} returns a feasible control action $u_{t+1}$ at the next time step $t+1$:

\textit{Case 1:} If $\bar x_{F, m^*_t} = x_{T}^0$ and the horizon $N_t  = 1$, then $k_{t+1}=T$, $N_{t+1}=1$ and $x_{t+1} = x_{T}^0$, therefore by the invariance of $\mathcal{G}$ for $m=0$ the FTOCP~$J(x_{t+1}, x_F, q_F, \mathcal{I}_{t+1}, N_{t+1})$ with $x_F = x_{k_{t+1}} = x_{T}^0$ is feasible.

\textit{Case 2:} If $\bar x_{F, m^*_t} = x_{T}^0$ and the horizon $N_t  > 1$, then $k_{t+1}=T$ and for $m=0$ the following state-input sequences
$[x_{t+1|t}^*, \ldots, x_{t+N_t|t}^*] \text{ and } [u_{t+1|t}^*, \ldots, u_{t+N_t-1|t}^*]$
are feasible for the FTOCP~$J(x_{t+1}, x_F, q_F, \mathcal{I}_{t+1}, N_{t+1})$ with $x_F = x_{k_{t+1}} = x_{T}^0$ as $x_{t+N_t|t}^* = x_{T}^0$ and $N_{t+1}=N_t-1$.

\textit{Case 3:} If $\bar x_{F, m^*_t} \neq x_{T}^0$, then $k_{t+1} = k_t + m^*_t+1$ and for $m=0$ the state-input sequences $[x_{t+1|t}^*, \ldots, x_{t+N_t|t}^*=x_{k_t+m_t^*}^0, x_{k_t+m_t^*+1}^0]$ and $[u_{t+1|t}^*, \ldots, u_{t+N_t-1|t}^*, u_{k_t+m_t^*}^0]$
are feasible for the FTOCP~$J(x_{t+1}, x_F, q_F, \mathcal{I}_{t+1}, N_{k+1})$ with $x_F = x_{k_{t+1}} = x_{k_t+m_t^*+1}$.

From \textit{Cases}~$1$--$3$, we have that if at time $t$ Algorithm~\ref{algo:policy} returns a feasible action~$u_t$, then 
at time $t+1$ Algorithm~\ref{algo:policy} returns a feasible control action~$u_{t+1}$. Now, we notice that at time $t=0$ the sequence of actions $[u_0^0, \ldots, u_{N-1}^0]$ is feasible for the FTOCP~$J(x_{0}, x_F, q_F, \mathcal{I}_{0}, N_0)$ with $x_F = x_{N_0}^0$, which in turns implies that Algorithm~\ref{algo:policy} returns a feasible action $u_t$ at all times and that state and input constraints are satisfied.

Finally, we show finite time convergence of the closed-loop system to the goal set $\mathcal{G}$. From Algorithm~\ref{algo:policy}, we have that $k_t$ increases at each time step until $k_t = T$ after at most $T-N$ time steps and, afterwards, that the horizon shrinks. Therefore, at time $T-1$ we have that $N_{T-1} = 1$, $t_F =k_{T-1}=T$ and $x_F=x_{T}^0$, thus the predicted state at time $T-1$ of the optimal trajectory from~\eqref{eq:optSeqProof} satisfies $x^*_{T|T-1} = x_{T}^0$, which in turns implies that $x_{T} = x_{T|T-1}^* = x_{T}^0\in\mathcal{G}$.
\end{proof} 

\begin{corollary}
Consider the closed-loop system~\eqref{eq:sys} and~\eqref{eq:policy}, where the policy~$\pi$ is given by Algorithm~\ref{algo:policy}.
Let Assumptions~\ref{ass:firstFeas}-\ref{ass:invariance} hold. If at time $t$ Algorithm~\ref{algo:policy} returns a feasible control action~$u_t \in \mathcal{U}$, then the closed-loop system~\eqref{eq:sys} and~\eqref{eq:policy} satisfies constraints~\eqref{eq:cnstr} and it converges to the goal set~$\mathcal{G}$.
\end{corollary}

The above corollary 
highlights the advantage of computing a policy that maps states to actions and it can be used to deal with perturbed initial conditions and uncertainties. In the result section we perform an empirical study where we test the robustness of the proposed methodology by changing initial conditions and simulating disturbances acting on the system.

\subsection{Iterative Improvement}
This section discusses the properties of the iterative Algorithm~\ref{algo:iterativeUpdate}. In particular, we show that at each policy update the cumulative cost associated with the closed-loop trajectories from~\eqref{eq:sys_j} is non-increasing. Notice that our strategy guarantees non-increasing cost at each update, but no guarantees are given about the optimality of the trajectory at convergence.

\begin{theorem}
For $i\in \{0,\ldots,j\}$ consider the closed-loop trajectories $\mathbf{x}^i$ from Algorithm~\ref{algo:iterativeUpdate}.
If Assumptions~\ref{ass:firstFeas}-\ref{ass:invariance} hold, then we have that at each policy update the cost associated with the closed-loop trajectories is non-increasing, i.e., $q_0^{i-1} \geq q_{0}^i,~\forall i \in \{1,\ldots,j\},$
where $q_0^i = \sum_{t=0}^{T-1} l(x_t^i,u_t^i)$.
\end{theorem}
\begin{proof} From Theorem~\ref{th:recFeas} it follows that at each $j$th update the policy~$\pi^j$ from Algorithm~\ref{algo:policy} returns a feasible action~$u_t^{j}$ and the feasible state-input trajectories $\mathbf{x}^j$ and $\mathbf{u}^j$.
At time $t$, let $c^{*,j}_t$ be the optimal cost of the $m_t^*$th FTOCP and let $[x_{t|t}^{*,j}, \ldots, x_{t+N_t|t}^{*,j}] \text{ and } [u_{t|t}^{*,j}, \ldots, u_{t+N_t-1|t}^{*,j}]$
be the optimal solution. Then we write the optimal cost as
\begin{equation}\label{eq:costWrittenOut}
\begin{aligned}
    c_t^{*,j} &= \sum_{k=t}^{t+N_t-1} l(x_{k|t}^{*,j}, u_{k|t}^{*,j}) + q_{k_t+m_t^*}^{j-1}\\
    &=l(x_{t|t}^{*,j}, u_{t|t}^{*,j})+\sum_{k=t+1}^{t+N_t-1} l(x_{k|t}^{*,j}, u_{k|t}^{*,j}) + q_{k_t+m_t^*}^{j-1}.
\end{aligned}
\end{equation}
Next, we consider three cases to analyze the time evolution of the optimal cost $c_t^{*,j}$ associated with the $m_t^*$th FTOCP solved at line~9 of Algorithm~\ref{algo:policy}:

\textit{Case 1:} If $\bar x_{F, m^*_t} = x_{T}^{j-1}$ and the horizon $N_t  = 1$, then $k_{t+1}=T, N_{t+1}=1$, $x_{t+1}=x_T^{j-1}$ and the FTOCP for $m=0$ is feasible (\textit{Case~1} of Theorem~\ref{th:recFeas}) which together with Assumption~\ref{ass:invariance} imply that $c_{t+1}^{*,j}=0$. Thus, we have that 
\begin{equation}\label{eq:th2Case1}
\begin{aligned}
    c_{t}^{*,j} &= l(x_{t|t}^{*,j}, u_{t|t}^{*,j}) + q_{m_t^*}^{j-1} = l(x_{t|t}^{*,j}, u_{t|t}^{*,j})+ c_{t+1}^{*,j},
    \end{aligned}
\end{equation}
as the terminal cost $q_{m_t^*}^{j-1} = 0 = c_{t+1}^{*,j}$ for $\bar x_{F, m^*_t} = x_{T}^{j-1}$. 

\textit{Case 2:} If $\bar x_{F, m^*_t} = x_{T}^{j-1}$ and the horizon $N_t > 1$, then we have that the last two terms in~\eqref{eq:costWrittenOut} represent the cost $\bar c_{t+1}^j$ associated with the state-input sequences $[x_{t+1|t}^{*,j},\ldots,x_{t+N_t|t}^{*,j}]$ and $[u_{t+1|t}^{*,j},\ldots,u_{t+N_t-1|t}^{*,j}]$, which are feasible at time $t+1$ (\textit{Case}~2 of Theorem~\ref{th:recFeas}) and therefore we have that
\begin{equation}\label{eq:th2Case2}
    c_{t}^{*, j}= l(x_{t|t}^{*,j}, u_{t|t}^{*,j}) + \bar c_{t+1}^j \geq l(x_{t|t}^{*,j}, u_{t|t}^{*,j}) + c_{t+1}^{*, j}.
\end{equation}

\textit{Case 3:} If $\bar x_{F, m^*_t} \neq x_{T}^{j-1}$, then by definition~\eqref{eq:realizedCost} we have that the optimal cost can be written as 
\begin{equation*}
\begin{aligned}
    c_t^{*,j} &=l(x_{t|t}^{*,j}, u_{t|t}^{*,j})+\sum_{k=t+1}^{t+N_t-1} l(x_{k|t}^{*,j}, u_{k|t}^{*,j}) + q_{k_t+m_t^*}^{j-1}\\
    &=l(x_{t|t}^{*,j}, u_{t|t}^{*,j})+\sum_{k=t+1}^{t+N_t-1} l(x_{k|t}^{*,j}, u_{k|t}^{*,j}) \\& \quad\quad\quad\quad\quad\quad\quad\quad+ l(x_{k_t+m_t^*}^{j-1}, u_{k_t+m_t^*}^{j-1})+q_{k_t+m_t^*+1}^{j-1}.
\end{aligned}
\end{equation*}
Notice that the last three terms in the above equation represent the open-loop cost $\bar c_{t+1}^j $ associated with the state-input sequences  $[x_{t+1|t}^{*,j}, \ldots, x_{k_t+m_t^*|t}^{j-1}, x_{k_t+m_t^*+1|t}^{j-1}]$ and $[u_{t+1|t}^{*,j}, \ldots, u_{t+N_t-1|t}^{*,j}, u_{k_t+m_t^*|t}^{j-1}]$, which are feasible at time $t+1$ (\textit{Case}~3 of Theorem~\ref{th:recFeas}) and therefore we have that
\begin{equation}\label{eq:th2Case3}
    c_{t}^{*,j} = l(x_{t|t}^{*,j}, u_{t|t}^{*,j}) + \bar c_{t+1}^j \geq  l(x_{t|t}^{*,j}, u_{t|t}^{*,j}) + c_{t+1}^{*,j}.
\end{equation}


Finally, we notice that $c_T^{j,*}=0$ as from Theorem~\ref{th:recFeas} $x_T^j\in\mathcal{G}$ and therefore from equations~\eqref{eq:th2Case1}--\eqref{eq:th2Case3} we have that 
\begin{equation*}
    c_0^{*,j} \geq l(x_1^j, u_1^j) + c_{2}^* \geq \ldots \geq \sum_{t=0}^{T-1} l(x_{t}^j, u_{t}^j) + c_T^{*,j} = q_0^j, 
\end{equation*}
as $x_{t|t}^{*,j}= x_t^j$ and $u_{t|t}^{*,j}= u_t^j$ for all $t \in \{0, \ldots,T\}$. Furthermore, at time $t=0$ and for $m=0$ the sequence of open-loop actions $[u_0^j, \ldots, u_{N-1}^j]$ is feasible for $J(x_{0}^j, x_F, q_F, \mathcal{I}_{0}, N_0)$ with $x_F = x_{N}^j$ and the associated cost is $q_0^{j-1}$, which in turns implies 
 $q_0^{j-1}\geq c_0^j \geq q_0^j$.
\end{proof} 

\begin{figure}[t!]
    \centering
	\includegraphics[width= 1.00\columnwidth,trim=4 4 4 4,clip]{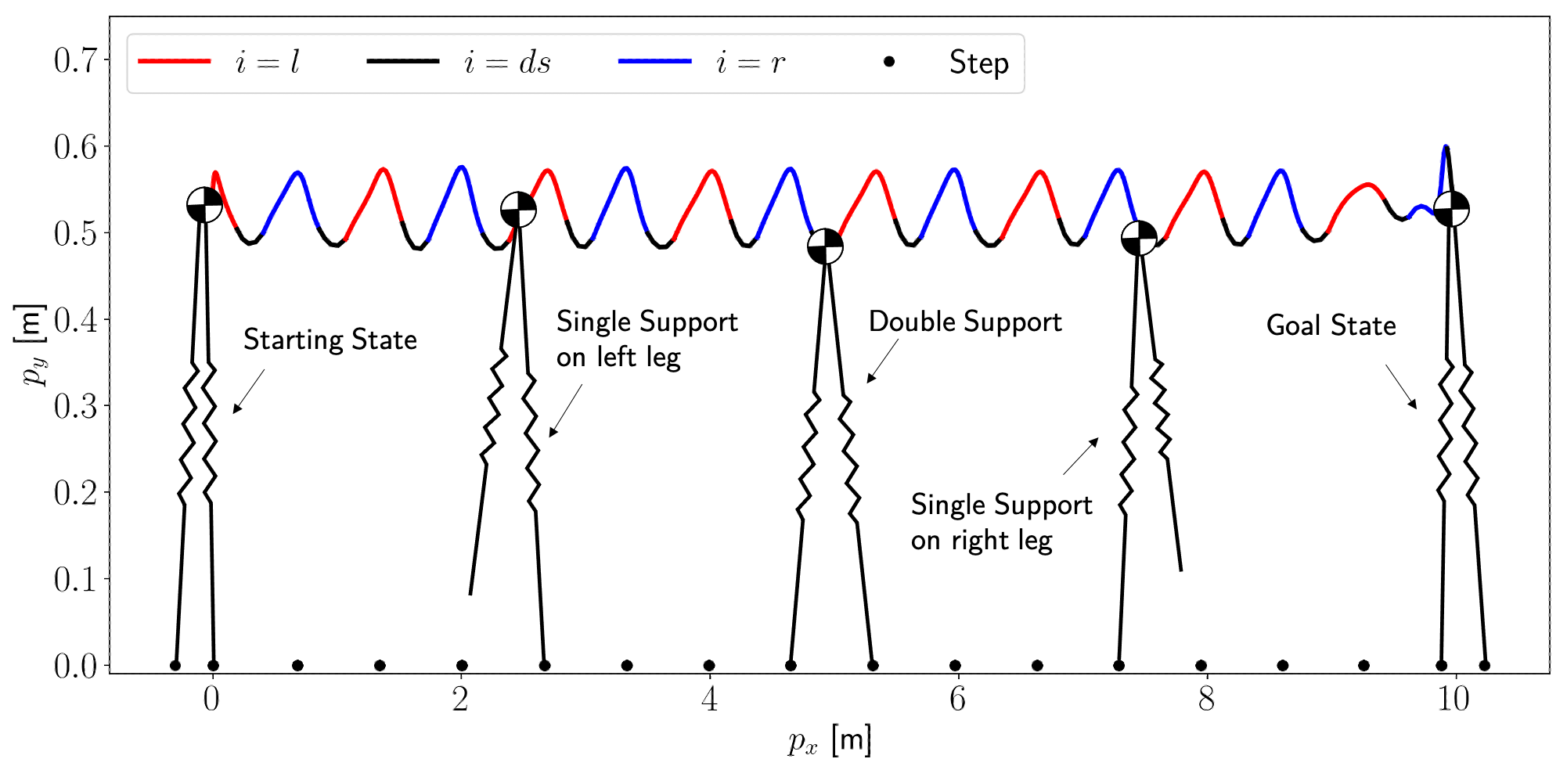}
    \vspace{-0.85cm}  \caption{Trajectory of the CoM and foot step locations associated with the feasible trajectory $\mathbf{x}^0$ used to initialize the control policy from Algorithm~\ref{algo:policy}.}
    \label{fig:first}
\end{figure}

\section{Application to SLIP Walking}
\subsection{Spring Loaded Inverted Pendulum Model}{}\label{sec:SLIP}
This section describes the Spring Loaded Inverted Pendulum (SLIP) model with massless legs~\cite{shahbazi2016unified}. The system state is $x = [p_x, p_y, v_x, v_y, z_x^l, z_x^r, z_y^l, z_y^r]$, where $(p_x, p_y)$ is the position of the Center of Mass (CoM), $(v_x, v_y)$ the velocity of the CoM, and $(z_x^k, z_y^k)$ the position of the $k$th foot for $k \in \{l,r\}$. The input vector $u = [\delta, v_z, v_y^l, v_y^r]$, where $\delta$ represents the change of leg stiffness, $v_z$ the velocity along the $x$-axis of the foot not in contact with the ground, and $(v_y^l, v_y^r)$ the velocities of the two feet on the $y$-axis. Given the state of the system, we can compute the angle $\theta^k = \text{tg}^{-1}( (p_x - z_x^k) / p_y ), \forall k \in \{l, r\}$ and the length $l^k = l_0 - \sqrt{(p_x - z_x^k )^2 + p_y^2} , \forall k \in \{l, r\}$, where $l_0=0.55$ is the rest leg length, for more details please refer to~\cite{shahbazi2016unified}.
These quantities are used to define the dynamics:
\begin{equation*}
\begin{aligned}
    f(x, u) &= x \\
    &+ \begin{bmatrix} v_x \\
    v_y \\
    (\delta+\delta_0)[\gamma^l l^l\sin(\theta^l)+\gamma^r l^r \cos(\theta^r)]\\
    (\delta+\delta_0)[\gamma^l l^l\cos(\theta^l)+\gamma^r l^r \cos(\theta^r)] - mg\\
    (1-\gamma^l) v_z\\
    (1-\gamma^r) v_z\\
    v_y^l \\ v_y^r
    \end{bmatrix}dt
\end{aligned}
\end{equation*}
where the integer variable $\gamma^k$ equals one if the $k$th foot is in contact with the ground. In particular, we have that $(\gamma^l, \gamma^r) = (1,1) \text{ if } x \in \mathcal{D}_{ds}$,  $(\gamma^l, \gamma^r) = (1,0) \text{ if } x \in \mathcal{D}_l$, and $(\gamma^l, \gamma^r) = (0,1)\text{ if } x \in \mathcal{D}_r$, where the regions are defined as follows:
\begin{equation*}
    \begin{aligned}
         \mathcal{D}_{ds} &= \{x \in \mathbb{R}^n | l^r \leq l_{\mathrm{max}}, l^l \leq l_{\mathrm{max}}, z_y^l = 0 \text{ and } z_y^r = 0\},\\
         \mathcal{D}_l &= \{x \in \mathbb{R}^n | l^r \leq l_{\mathrm{max}}, z_y^l = 0 \text{ and } z_y^r > 0\},\\
         \mathcal{D}_r &= \{x \in \mathbb{R}^n | l^l \leq l_{\mathrm{max}}, z_y^l > 0 \text{ and } z_y^r = 0\}.
    \end{aligned}
\end{equation*} 
Finally, we notice that when a sequence of regions is fixed, the vertical motion of the feet can be computed independently from the other states.

\vspace{-0.2cm}

\subsection{Simulations}\label{sec:Results}
First, we compute a feasible trajectory that steers the system from standing still at $(p_x, p_y)=(0, 0.85)$ to a goal state $x^g = (10, 0.85, 0, 0, 9.9, 10.2, 0, 0)$. In order to compute a feasible trajectory $\mathbf{x}^0$, we fixed a sequence of regions $\{\mathcal{D}_t\}_{t=0}^T$ where the system should be at each time $t$ and we solved the resulting NLP with IPOPT~\cite{wachter2006implementation} using CASADI~\cite{andersson2019casadi}. Furthermore, we added a slack variable to the terminal constraint, we used 
\begin{equation}\label{eq:expCost}
\begin{aligned}
    l(x,u)&= ||p_x-p_x^{\textrm{g}}||_2^2 + 10 ||p_y-p_y^{\textrm{g}}||_2^2 + ||v_x||_2^2 + ||v_y||_2^2 \\ &+ ||\delta||_2^2 + 0.1||v_z||_2^2,  
\end{aligned}
\end{equation}
and the input constraints $\delta \in \{\bar \delta : ||\bar \delta|| \leq 10 \}$ and $v_z \in \{\bar v_z : ||\bar v_z|| \leq 10 \}$. Code available at~\texttt{\url{https://github.com/urosolia/SLIP}}.

\begin{figure}[t!]
    \centering
	\includegraphics[width= 1.00\columnwidth,trim=4 4 4 4,clip]{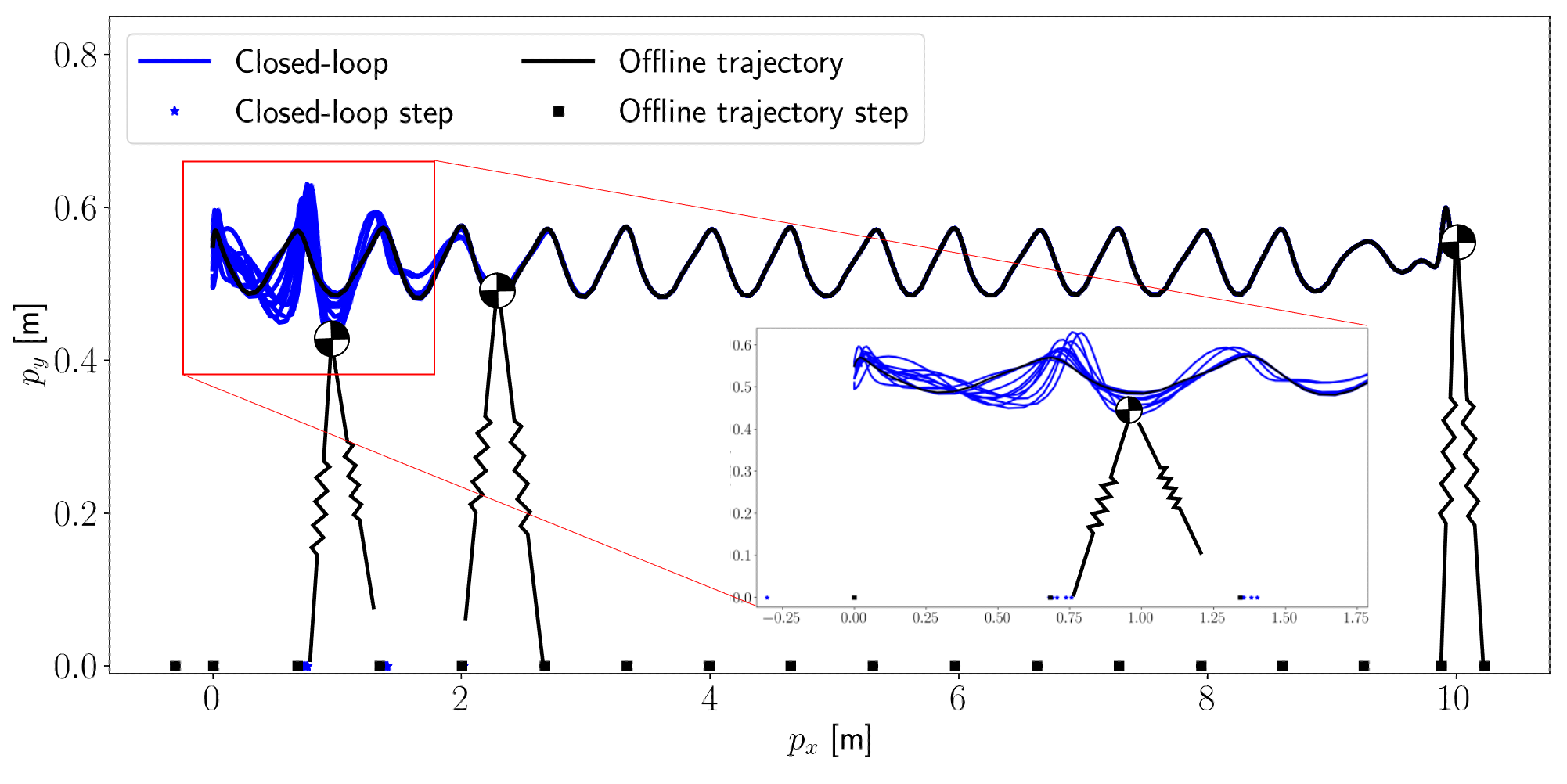}
    \vspace{-0.85cm}  
    \caption{Closed-loop trajectories of the CoM and foot steps for different initial conditions.}
    \label{fig:diffIC}
\end{figure}

Figure~\ref{fig:first} shows the feasible trajectory which steers the system from the starting configuration to the goal position while transitioning through phases of double support (black) and single support with the left (red) and right (blue) legs. The figure also shows the locations where the feet are in contact with the ground. This feasible trajectory is used to initialize the control policy from Algorithm~\ref{algo:policy}. We tested the proposed strategy with $N=30$ on $10$ different initial conditions in the neighborhood of the starting state $x_S$. We set $M=1$ to test the robustness to disturbances and changes in initial conditions of Algorithm~\ref{algo:policy}, when only one FTOCP is solve at each time~$t$. Figure~\ref{fig:diffIC} shows that for all initial conditions the controller is able to steer the system to the goal state. Notice that in order to stabilize the system the proposed strategy is able to plan a sequence of foot steps, which are different from the one associated with the feasible trajectory. Finally, it takes on average less than $0.05s$ to run Algorithm~\ref{algo:policy} with a maximum computation time of $0.114s$ across all simulations.

Furthermore, we compared the proposed strategy with a tracking controller which is defined removing the terminal constraint from~\eqref{eq:ftocp} and using a tracking cost instead of~\eqref{eq:expCost}. Both our method and the tracking controller are able complete the task. However, as shown in Figure~\ref{fig:disturbance}, the tracking controller fails to reach the goal when a disturbance hits the system. In Figure~\ref{fig:disturbance}, it is interesting to notice that the proposed approach initially deviates more from the offline trajectory compared to the tracking controller. This deviation allows the controller to compensate for the disturbance and to stabilize the system back to a periodic gait.

Finally, we tested Algorithm~\ref{algo:iterativeUpdate} to iteratively update the control policy. We set $j=40$ and we changed the stage cost to encode the objective of steering the system from the starting state to the goal state in minimum time. In particular, we defined the stage cost $\bar l(x,u) = \mathds{1}_\mathcal{G}(x) + 0.0001 l(x,u)$, where the function $\mathds{1}_\mathcal{G}(x)=0$ if $x \in x_{T}^0$ and $\mathds{1}_\mathcal{G}(x)=1$ otherwise. Moreover, we set $M=40$ to allow Algorithm~\ref{algo:policy} to solve multiple instances of the FTOCP~\eqref{eq:ftocp}. Note that as $M$ gets larger the controller can better explore the state space, but as a trade off the computational cost increases. In this example it takes on average less than $0.4s$ and at most $0.9s$ to run Algorithm~\ref{algo:policy}.
We initialized the proposed policy iteration strategy with the feasible trajectory $\mathbf{x}^0$ from Figure~\ref{fig:first} that steers the system from the starting point to the goal state in $396$ time steps, and our algorithm returned a policy which completes the task in $193$ time steps. The closed-loop trajectory is shown in Figure~\ref{fig:x40}. First the controller accelerates and as a result the CoM oscillates more compared to the first feasible trajectory from Figure~\ref{fig:first}. Finally, the controller slows down and reduces the oscillation to reach the goal position with zero speed and two feet on the ground.

\vspace{-0.1cm}

\section{Conclusions}
We presented an algorithm to synthesize control policies for hybrid systems by leveraging a feasible trajectory for the control task. 
We showed that the proposed methodology guarantees constraints satisfaction and convergence in finite time. Building upon the proposed synthesis strategy, we presented a policy iteration algorithm which guarantees that at each policy update the closed-loop performance is non-decreasing. Finally, we tested the proposed strategy on a discretize SLIP model.
\vspace{-0.15cm}

\section{Acknowledgements}
The authors would like to thank  
Andrew J. Taylor, Noel Csomay-Shanklin, and Wenlong Ma for suggestions and proofreading the manuscript.

\begin{figure}[t!]
    \centering
	\includegraphics[width= 1.00\columnwidth,trim=4 4 4 4,clip]{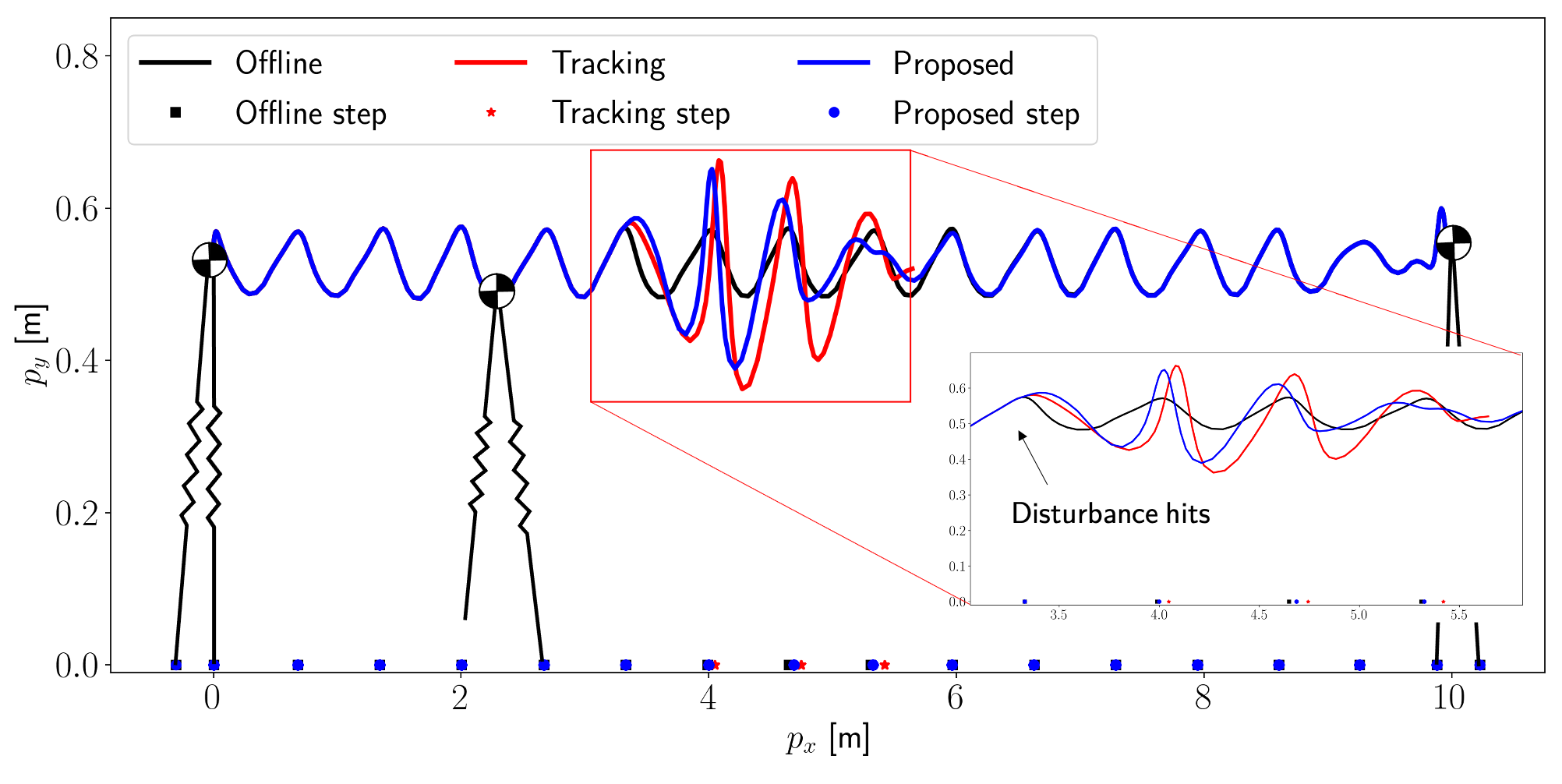}
	\vspace{-0.85cm}  
    \caption{Comparison between an MPC tracking controller and the proposed strategy, when a disturbance hits the system.}
    \label{fig:disturbance}
\end{figure}

\vspace{-0.15cm}

\renewcommand{\baselinestretch}{0.99}

\bibliographystyle{IEEEtran} 
\bibliography{IEEEabrv,mybibfile}

\begin{figure}[t!]
    \centering
	\includegraphics[width= 1.00\columnwidth,trim=4 4 4 4,clip]{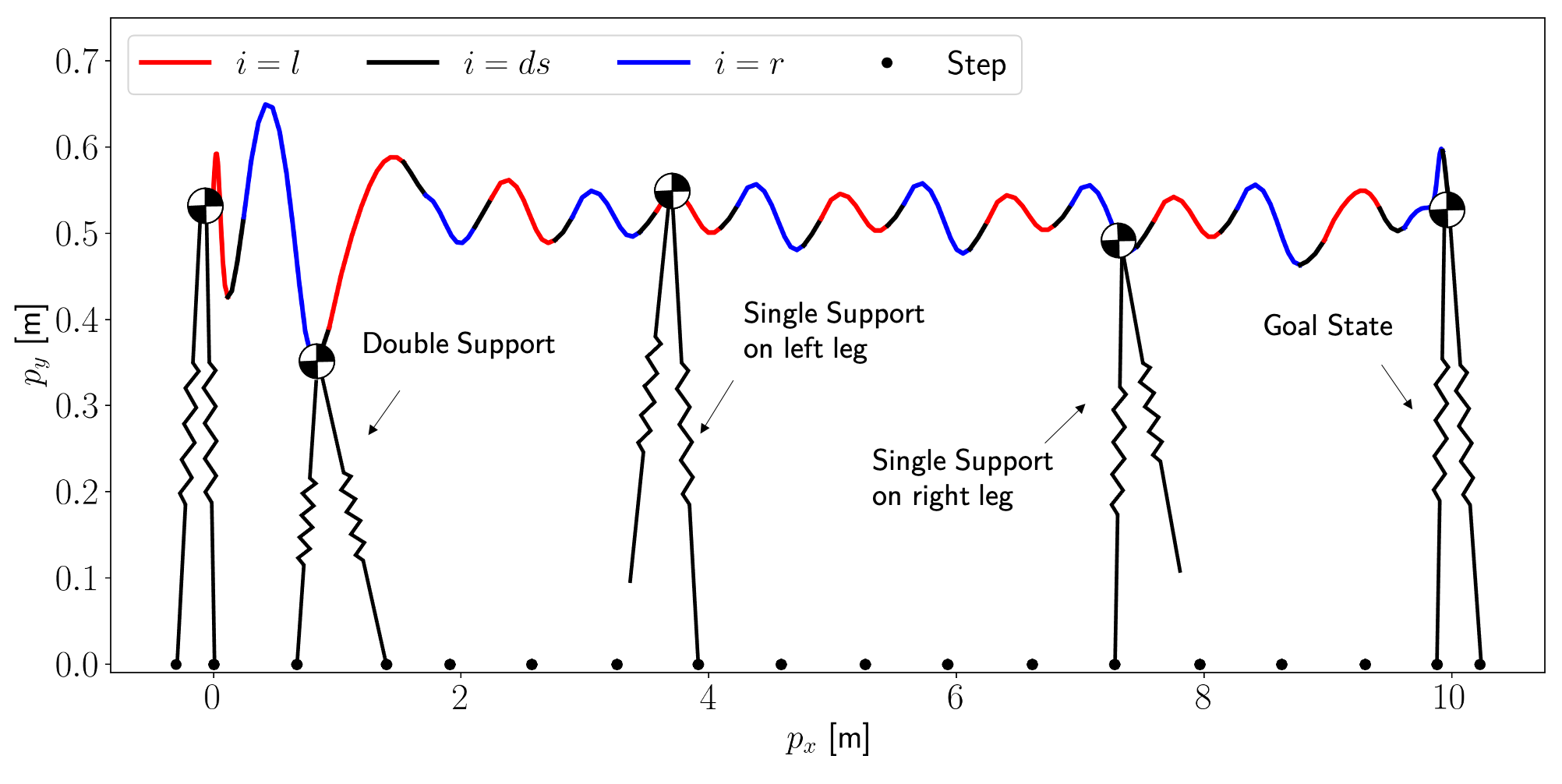}
	\vspace{-0.85cm}  
    \caption{Closed-loop trajectory computed from Algorithm~\ref{algo:iterativeUpdate} after $40$ updates of the control policy.}
    \label{fig:x40}
\end{figure}
\end{document}